\documentclass[11pt,a4paper,english]{article} 
\usepackage{authblk}

\title{Stationary Spacetimes and
	Self-Adjointness in Klein-Gordon Theory}
\author[1]{Felix Finster\footnote{finster@ur.de}}
\author[2]{Albert Much\footnote{amuch@matmor.unam.mx}}
\author[3]{Robert Oeckl\footnote{robert@matmor.unam.mx}}
\affil[1]{Fakult\"at f\"ur Mathematik\\
	Universit\"at Regensburg\\
	D-93040 Regensburg}
\affil[2]{Institut f\"ur Theoretische Physik\\ Universit\"at Leipzig\\ D-04103 Leipzig} 
\affil[2,3]{Centro de Ciencias Matem\'aticas\\
	Universidad Nacional Aut\'onoma de M\'exico\\
	C.P. 58190, Morelia, Michoac\'an, Mexico}
\date{October 2019}	    
\usepackage{float}            
\usepackage[T1]{fontenc} 
\usepackage[cmex10]{amsmath}  
\usepackage{amstext}          
\usepackage{mathrsfs}
\usepackage{amsfonts}         
\usepackage{amssymb}          
\usepackage{amsthm}    
\usepackage{setspace}
\usepackage{bm}               
\usepackage{enumerate}        
\usepackage[bottom]{footmisc} 
\usepackage{array}            
\usepackage{textcomp}         
\usepackage{pdfpages}         
\usepackage{parskip}          
\usepackage[right]{eurosym}   
\usepackage{xcolor}           
\usepackage[square,numbers]{natbib}	  
\usepackage{tikz}
\usetikzlibrary{matrix}

\usepackage{amsmath}  
\usepackage[colorlinks,pdfpagelabels,pdfstartview = FitH,bookmarksopen = true,bookmarksnumbered =
true,linkcolor = black,plainpages = false,hypertexnames = false,citecolor = black]{hyperref}
\usepackage[a4paper,textwidth=17cm,textheight=24cm,headheight=14pt]{geometry}
\allowdisplaybreaks

\newtheorem{theorem}{\textsc{Theorem}}[section]
\newtheorem{lemma}[theorem]{\textsc{Lemma}}
\newtheorem{proposition}[theorem]{\textsc{Proposition}}

\newtheorem{corollary}[theorem]{\textsc{}Corollary}
\newtheorem{definition}[theorem]{\textsc{Definition}}

\newtheorem{convention}[theorem]{Conventions}

\newtheorem{assumption}[theorem]{Assumption}
\newcommand{\R}{\mathbb{R}}

\newcommand{\Z}{\mathbb{Z}}

\usepackage{fancyhdr}
\numberwithin{equation}{section} 
\begin{document}
	\maketitle
	\abstract{
		We consider the problem of essential self-adjointness of the spatial part of the Klein-Gordon operator in stationary spacetimes. This operator is shown to be a Laplace-Beltrami type operator plus a potential.
		In globally hyperbolic spacetimes, essential self-adjointness is proven assuming smoothness of the metric components and semi-boundedness of the potential. This extends a recent result for static spacetimes to the stationary case. Furthermore, we generalize the results to certain non-globally hyperbolic spacetimes.} 
	
	\tableofcontents
	
	\section{Introduction}
	In a recent paper~\cite{MO} essential self-adjointness of the spatial part of the Klein-Gordon operator (with potential)
	was proved under fairly mild assumptions. In particular, this proof applies in all smooth globally hyperbolic static cases. Although any globally hyperbolic manifold $(M,g)$  admits a global time function
	for which the metric is of the form (see~\cite[Theorem 1.1]{B05}),
	\begin{equation}
	g=-N^2 dt^2+g_{ij}dx^idx^j,
	\label{m}\end{equation}
	in many applications one must deal with more general coordinates with a non-vanishing shift vector. This means
	that there is a term mixing the temporal and spatial components, represented by the metric
	\begin{align*}
	g=\tilde{N}^2\,dt^2+2N_i\,dt\,dx^i+{g}_{ij} \,dx^i\,dx^j.
	\end{align*}  
	In general all components can be time dependent. However, in this paper we restrict attention to {\em{time independent}} metrics. One important example of such stationary spacetimes is the Kerr metric
	(see for example~\cite{oneill, chandra}). 
	Hence, one goal of this paper is to generalize the previously mentioned
	result (a special case of~\cite{MO}; see also \cite[Theorem 7.1]{K78})
	to the case of globally hyperbolic \textit{stationary} spacetimes.
	Apart from the importance of this result for the Kerr metric and stability considerations,
	it is also useful for computing various propagators in quantum field theory; see for example~\cite{DS17}. In particular, the essential self-adjointness of the spatial part of the Klein-Gordon operator for globally hyperbolic spacetimes is the first assumption that the authors in~\cite[Assumption 1.a.]{DS17} make. Moreover, a proof of essential self-adjointness of the respective operator can be used to rigorously construct complex structures on the classical solution spaces
	in stationary spacetimes in the context of  quantum field theory, as done in \cite{K78} and  \cite{MO2}. 
	
	Our main result (see Theorem~\ref{mt}) generalizes the globally hyperbolic static case to the stationary one,
	without any further restrictions, except for smoothness of the manifold.
	The techniques also apply to spacetimes with ergo regions like the Kerr geometry
	(see Theorem~\ref{them:mtad}). Finally, we adapt the methods to obtain similar results
	for certain non-globally hyperbolic stationary spacetimes (see Theorem~\ref{thm54}).
	
	Our results lead to the following directions of future research. First, it would be useful (at least in the context of \cite{DS17}) to prove that the current result can be generalized to time-dependent metrics. Preliminary results indicate that the generalization in this direction is straightforward. Second, the dynamics can be formulated on stably causal,
	non-globally hyperbolic manifolds that are static or stationary, as long as the spatial part of the operator of the Klein-Gordon equation is self-adjoint, see \cite{NGH2}, \cite{NGH3}, \cite{NGH4}. Hence proving \textit{essential} self-adjointness would choose a unique self-adjoint extension of the respective operator. It is also of interest to explore
	if the resulting complex structures on the solution spaces
	in the non-globally hyperbolic stationary case match the framework of \cite{NGH1}. This is current work in progress. 
	\begin{convention}
		Throughout this work we use Greek letters $\mu, \,\nu=0,\ldots,3$ for spacetime indices and we use Latin letters  $i,\,j,\,k,\ldots$ for  spatial components which run from $1,\ldots,3$.	Moreover, we use the  signature $(-1,1,1,1)$, bold letters for the spatial metrics and $|g|$ denotes the determinant w.r.t.\ the metric $g$. 
	\end{convention}

	\section{The Geometry}
	
	In   what follows, we restrict our attention to stationary    spacetimes. A spacetime $(M,g)$ is called stationary if there is a Killing vector field which is timelike near infinity (see~\cite[Page 128]{FU}). This means in appropriate charts that the metric components $g_{\mu\nu}$ do not depend on the time coordinate~$t$.
	Furthermore, we consider   a four-dimensional  smooth spacetime manifold  $M$ that is the product manifold $M=\mathbb{R}\times\Sigma$, where   $\Sigma$ is a  smooth spacelike hypersurface,  (strictly speaking $\Sigma_t=\{t\}\times\Sigma$)  
	with induced Riemannian metric ${g}_{ij}(\vec{x})$, and   metric $g$   of the form 
	\begin{align}\label{eq:metric1b}  
	g=(-N^{2}(\vec{x}) +N_i N^i )\,dt^2+2N_i(\vec{x})\,dt\,dx^i+{g}_{ij}(\vec{x}) \,dx^i\,dx^j ,
	\end{align}  
	with $N\in C^{\infty}(\Sigma,\mathbb{R}^+)$ being a non-negative function  and a vector field $N^i\partial_i\in\Gamma(M,T\Sigma)$ tangential to $\Sigma_t$.    The components of the inverse metric ${g}^{\mu\nu}$ are given by 
	\begin{align*}  
	g^{00}=-N^{-2},\qquad  g^{0i}= N^{-2}N^i,\qquad {h}^{ij}= {g}^{ij} -N^{-2}N^iN^j.
	\end{align*}  
	The inverse of the metric ${h}^{ij}$, which is essential in the following sections, is given by  
	\begin{align*}
	{h}_{ij}=  {g}_{ij} +\frac{N_i N_j}{N^{2}-N_iN^i}.
	\end{align*}   \begin{assumption}\label{assump:met}
		In what follows, we assume that  the   shift vector field $\vec{N}$   is not completely independent but is restricted by the  condition 
		\begin{equation}\label{eq:c}
		-N^{2}+N_iN^i<0,
		\end{equation}
		unless otherwise stated.  This condition is equivalent to the assumption that the Killing vector field~$X(\equiv\partial_t)=N\,n(\Sigma)+N^i\partial_{i}$ is everywhere timelike (where~$n(\Sigma)$ is a unit future-pointing normal vector field on~$\Sigma$; see also~\cite[Equation 1.3]{K78}).  
	\end{assumption} 
	
	\section{The Main Result}
	
	Next, we investigate the spatial part of the Klein-Gordon equation
	\begin{align*}  
	\left(({\sqrt{|g|}})^{-1}\partial_{\mu}(\sqrt{|g|}g^{{\mu}{\nu}}\partial_{\nu})-m^2(\vec{x})\right)\Phi =0,
	\end{align*}
	in a stationary spacetime $(M,g)$ with metric tensor given in Equation (\ref{eq:metric1b}),
	where~$m(\vec{x})$ is a potential which may include the rest mass, scalar curvature and possibly
	a time-independent external potential.
	\begin{proposition}\label{p1}
		The spatial part of the Klein-Gordon equation resulting from  the metric $g$, given in Equation (\ref{eq:metric1b}),  has the form of a scaled weighted Laplace operator (see Definition and Equation \ref{lapop}) plus a potential term, i.e.,
		\begin{align}\nonumber
		w^2& =		- N^2\, \frac{1}{\rho\sqrt{ |\textbf{h}|}}  \partial_i\left(   \rho\,\sqrt{ |\textbf{h}|}   h^{ij}     \partial_j \right)+V(x)
		\\&=-  N^2\, \Delta_{\mu,\textbf{h}}+V(x),\label{op}
		\end{align}
		where we use the following definition of variables,
		\begin{equation}\label{ro}
		\rho:= \sqrt{|g|}  (\sqrt{ |\textbf{h}|})^{-1},  \qquad\qquad V(x):=  N^2\, m^2(x).
		\end{equation} 
	\end{proposition}
	\begin{proof}
		See Appendix \ref{pp1}.
	\end{proof}
The functional properties of the function  $\rho$ are given in the following result.
	\begin{lemma}\label{lem:rho}
		The function $\rho=\sqrt{|g|}  (\sqrt{ |\textbf{h}|})^{-1}$ is explicitly given by 
		\begin{align*}
		\rho=\sqrt{ |g_{00}^{-1}|} .
		\end{align*}
		Furthermore, it is a strictly positive and smooth function.
	\end{lemma}\begin{proof} See Appendix \ref{plem:rho}.
	\end{proof}
	In analogy to the proof in \cite{MO}, for the proof of essential self-adjointness of the spatial part of the Klein-Gordon equation we use the theory of weighted Hilbert spaces and the related fundamental theorem given in \cite{S01} (see Appendix \ref{thm:sch}).
	We proceed by supplying a minimal introduction. For further details, see the excellent reference \cite{AG1}. We begin with the definition of weighted manifolds and the corresponding weighted Laplace-Beltrami operator \cite[Chapter 3.6, Definition 3.17]{AG1}. 
	\begin{definition}\label{def21}
		A triple $(\Sigma,	\mathbf{h},\mu)$ is called a \textbf{weighted manifold}, if $(\Sigma,	\mathbf{h})$ is a Riemannian manifold and $\mu$ is a measure on $\Sigma$ with a smooth and everywhere positive density function $\rho$. The \textbf{weighted Hilbert space}, denoted as $L^2(\Sigma, \mu)$, is given  as the space of all  square-integrable functions  on the manifold $\Sigma$ with respect to the measure $\mu$.  	The corresponding  \emph{weighted Laplace-Beltrami operator} (also called the Dirichlet-Laplace operator),  denoted by $\Delta_{\mu,	\mathbf{h}}$ is given by 
		\begin{equation}\label{lapop}
		\Delta_{\mu,\mathbf{h}}=\frac{1}{\rho \sqrt{|\mathbf{h}|}}\partial_i(\rho\sqrt{|\mathbf{h}|}h^{ij}\partial_j).
		\end{equation} 
	\end{definition}
	Another fact that is useful for our proof about weighted manifolds and weighted Laplace-Beltrami operators is given in the following  proposition  \cite[Chapter 3, Exercise 3.11]{AG1}. 
	\begin{proposition}\label{p5}
		Let $\alpha$ be  a smooth and everywhere positive function on a weighted manifold $(\Sigma,\mathbf{h},\mu)$ and define a new metric $\tilde{\mathbf{h}}$ and measure $\tilde{\mu}$ by $$\mathbf{\tilde{h}}=\alpha\, {\mathbf{h}},\qquad \mathrm{and} \qquad d\tilde{\mu}=\alpha\, d{\mu}.$$ Then, the weighted Laplace-Beltrami operator $\tilde{\Delta}_{ \tilde{\mu}}$ of the weighted manifold $(\Sigma, \mathbf{\tilde{h}}, \tilde{\mu})$ is given by 	$$\tilde{\Delta}_{ \tilde{\mu},\tilde{\mathbf{h}}}=\frac{1}{\alpha}{\Delta}_{ {\mu},{\mathbf{h}}}.$$ 
	\end{proposition}
	\begin{proof} 
		See \cite[Proof of Proposition~3.1]{MO}.
	\end{proof}
	Next, let us consider  the weighted manifolds $(\Sigma,\mathbf{h},\mu)$ and $(\Sigma, \mathbf{\tilde{h}}, \tilde{\mu})$, where $d\tilde{\mu}= N^{-2}\,d\mu$ and $\mathbf{\tilde{h}}=N^{-2}\,\mathbf{h}$. 
	Applying the last proposition  we rewrite the operator $w^2$ as a weighted Laplace-Beltrami operator plus a potential, i.e.,
	$$w^2=-\tilde{\Delta}_{ \tilde{\mu},\tilde{\mathbf{h}}}+V,$$
	and by using  \cite[Theorem~1.1]{S01} (see Appendix \ref{thm:sch})  we obtain the following theorem. 
	
	\begin{theorem}\label{mt} Let the Riemannian manifold $(\Sigma,\tilde{\mathbf{h}})$ be geodesically complete and let the	potential $V\in L^2_{loc}(\Sigma, \tilde{\mu})$ be such that it can be written as $V = V_+ + V_-$, where $V_+\in L^2_{loc}(\Sigma , \tilde{\mu})\geq 0$ and $V_-\in L^2_{loc}(\Sigma , \tilde{\mu})\leq 0$
		point-wise.  Moreover, let the operator    $w^2$ (from Equation~\eqref{op}) be semi-bounded from below. Then, the operator  $w^2$ is an essentially self-adjoint operator on $C_0^{\infty}(\Sigma )\subset L^{2}(\Sigma,  \tilde{\mu})$. 
	\end{theorem}
	
	\begin{proof}
		The proof is analogous to \cite[Theorem~4.1]{MO} and  follows readily from   Lemma \ref{lem:rho},  Proposition \ref{p5} and Theorem \ref{T5}.
	\end{proof}

	\section{Geodesic Completeness}
	
	In this section we assume that the potential $V$ satisfies all the requirements that are demanded in Theorem \ref{mt}.   The essential self-adjointness of the operator $w^2$ for a manifold $(M=\R\times\Sigma,g)$ is proven in Theorem \ref{mt} for the case that the Riemannian manifold $(\Sigma,\mathbf{\tilde{h}})$ is geodesically complete. Hence, proving geodesic completeness for the aforementioned Riemannian manifold with metric $\tilde{h}_{ij}=N^{-2}{h}_{ij}$ implies essential self-adjointness.
	
	In the context of geodesic completeness (of the respective Riemannian manifold) we use a fundamental theorem  for general time-dependent metrics of the Form (\ref{eq:metric}) given in \cite[Theorem~2.1]{ghc2} (see also \cite{ghc1,ghc3}). Let the spacetime $M$ be given by the product $M=\mathbb{R}\times\Sigma$, where $\Sigma$ is an $n$-dimensional smooth Riemannian manifold. We endow the manifold  $M$  with a $n+1$-dimensional Lorentz  metric $g$ of the form,
	\begin{align}\label{eq:metric}
	g=-N^2(\vec{x},t)dt^2+g_{ij}(\vec{x},t)\,(dx^i+N^i(\vec{x},t)\,dt)\,(dx^j+N^j(\vec{x},t)\,dt) .
	\end{align}
	Here, $N(\vec{x},t)$ is the lapse function, $N^i(\vec{x},t)$ the shift vector and the spatial slices $\Sigma_t=\{t\}\times\Sigma$ of $M$, are spacelike sub-manifolds equipped with the time-dependent metric $g_t=g_{ij}(\vec{x},t)dx^i\,dx^j$.\footnote{Such a product space $M$ is called a \textit{sliced space}. The spacetime is time-oriented by increasing $t$.} Moreover, let the following assumption be met. 
	\begin{assumption} \label{ass:bound} The components of the metric $g$ (Equation \eqref{eq:metric}) have the following bounds.\newline
		\begin{enumerate}
			\item The lapse function is bounded from above and below for all $t$, i.e.,
			$$0<\alpha_{B}\leq N(\vec{x},t)\leq\alpha_{C} ,$$
			where $\alpha_{B}, \alpha_{C}\in\mathbb{R}$.\newline
			\item The scalar product of the shift vectors $N^i(\vec{x},t)$ for the spatial metric is uniformly bounded from above by a number $B$.\newline
			\item The metric $g_{ij}(\vec{x},t)$ is uniformly bounded by the metric $g_{ij}(\vec{x},0)$ for all $t\in\mathbb{R}$ and tangent vectors $u\in T\Sigma$. That is, there exist constants $A,D\in\mathbb{R}>0$ such that
			\begin{align}\label{ineq:metAD}
			A\,g_{ij}(\vec{x},0)u^i\,u^j\leq g_{ij}(\vec{x},t)u^i\,u^j\leq D\,g_{ij}(\vec{x},0)u^i\,u^j.
			\end{align}
		\end{enumerate}
		
	\end{assumption}
	
	\begin{theorem}\label{thm:ghst} 
		For a spacetime manifold $(M,g)$, with metric $g$ that satisfies Assumption~\ref{ass:bound},  the following two statements are equivalent. \newline
		\begin{enumerate}
			\item $(\Sigma, \mathbf{g})$ is a complete Riemannian manifold. \newline
			\item The spacetime $(M,g)$ is globally hyperbolic. 
		\end{enumerate}
	\end{theorem}
	
	\subsection{Globally Hyperbolic Stationary Spacetimes}
	
	In \cite{MO} we proved the self-adjointness of $w^2$  for a large class of globally hyperbolic spacetimes\footnote{With smoothness requirements on $g_{00}$ and semi-boundedness of the potential.} including all static ones with metric of the form,
	\begin{align*}
	g=-N^2(\vec{x})dt^2+g_{ij}(\vec{x})\,dx^i\,dx^j.
	\end{align*}
	The proof relied on the fact that conformal transformations do not change the causal structure and on a theorem concerning ultra-static spacetimes given in \cite{K78}. Next, we prove the essential self-adjointness of $w^2$  for all globally hyperbolic stationary spacetimes  of the form
	\begin{align*}
	g=-N^2(\vec{x})dt^2+g_{ij}(\vec{x})\,(dx^i+N^i\,dt)\,(dx^j+N^j\,dt),
	\end{align*}
	by proving that  $(\Sigma,\tilde{\mathbf{h}})$ is geodesically complete. 
	Hence, we turn our attention to   the case where   the lapse function, the shift vector and the metric are time-independent. Moreover, our focus is to study the geodesic completeness of the conformally transformed metric (where the conformal factor is $N^{-2}$), since this is the remaining ingredient for the proof of essential self-adjointness.
	\begin{theorem}\label{thm:51}
		Let $(M,g)$ be a stationary globally hyperbolic spacetime of the form given in Equation (\ref{eq:metric}). Then, the Riemannian manifold with the conformally transformed metric $(\Sigma, \mathbf{\tilde{h}})$ is complete.
	\end{theorem}
	\begin{proof}
		The conformally transformed stationary spacetime $(M,\tilde{g})$, given by a metric of the form 
		\begin{align*}
		\tilde{g}=-dt^2+\tilde{g}_{ij}(\vec{x})\,(dx^i+N^i\,dt)\,(dx^j+N^j\,dt),
		\end{align*}
		with $\tilde{\mathbf{g}}=N^{-2}\mathbf{g}$, is globally hyperbolic since conformal transformations do not change the causal structure (see \cite[Appendix D]{WA}). Moreover, the conformally transformed metric satisfies the requirements of Theorem \ref{thm:ghst}; The   Lapse function is equal to one and therefore  bounded (possesses a lower and upper bound), and the scalar products of the conformally transformed spatial metric shift vectors   are also bounded from above. This follows easily from Assumption \ref{assump:met}, i.e., by   Inequality~\eqref{eq:c},
		from which we have,
		\begin{align*}
		-1+N^{-2}g_{ij}N^iN^j=-1+\tilde{g}_{ij}N^iN^j<0.
		\end{align*}
		The needed inequality follows, i.e.,
		\begin{align*}
		\tilde{g}_{ij}N^iN^j<1.
		\end{align*}
		Since the conformally transformed spatial metric  is time-independent it is uniformly bounded by itself. That is,
		\begin{align*}
		A\,\tilde{g}_{ij}u^i\,u^j\leq \tilde{g}_{ij}(\vec{x} )u^i\,u^j\leq D\,\tilde{g}_{ij}u^i\,u^j, 
		\end{align*}
		for constants $0<A\leq1$ and $D\geq1$ and tangent vectors $u\in T\Sigma$. Henceforth, we conclude from the global hyperbolicity of $(M,\tilde{g})$ and Theorem \ref{thm:ghst} that $(\Sigma, \mathbf{\tilde{g}})$ is a complete Riemannian manifold.  Since, the metric $\mathbf{\tilde{h}}$ is positive and the difference $\mathbf{\tilde{h}}-\mathbf{\tilde{g}}$ is positive semi-definite at every point,  it follows that the  Riemannian manifold $\Sigma$ with the metric $\mathbf{\tilde{h}}$ is complete \cite[Theorem~1, Remark~2]{WG}.
	\end{proof}
	Hence, by this theorem, essential self-adjointness of the operator $w^2$ (from Equation~\ref{op}) holds for  \textbf{all globally hyperbolic stationary spacetimes} that are smooth (and with suitable requirements on the potential, see \cite[Theorem~1.1]{S01}). This theorem is the generalization of a previous result, see \cite{MO}, to globally hyperbolic {\em{stationary}} spacetimes.
	
	\subsection{Example: Kerr Metric}
	
	An interesting and intensely studied model in the context of stationary spacetimes is the Kerr metric
	which describes a rotating, stationary, axially symmetric black hole.
	In this section we explain how our methods can be adapted such as to apply
	to the spatial Klein-Gordon operator in the Kerr geometry. The reason why such an adaptation is necessary
	is that our above proof of essential self-adjointness relied on the Assumption~\ref{assump:met}, 
	\begin{equation*} 
	-N^{2}+N_iN^i<0. 
	\end{equation*}
	If this term is  positive or equal to zero, our main Theorem~\ref{mt} does not apply.
	This is indeed  the case for the Kerr spacetime $(M,g_\text{Kerr})$ in Boyer-Lindquist coordinates.	Using the notation in~\cite{FF03}, we write the metric in Boyer-Lindquist coordinates~$(t,r,\vartheta,\varphi)$
	with $r>0$, $0\leq\vartheta\leq\pi$, $0\leq\varphi\leq2\pi$ as
	\begin{equation}\label{eq:metker}
	g_\text{Kerr}=-\frac{\Delta}{U}(dt-a\,\sin^2\vartheta\,d\varphi)^2+U\left(\frac{dr^2}{\Delta}+d\vartheta^2\right)+\frac{\sin^2\vartheta}{U}(a\,dt-(r^2+a^2)d\varphi)^2,
	\end{equation}
	with 
	\begin{align*}
	U(r,\vartheta)=r^2+a^2\cos^2\vartheta,\qquad \qquad \Delta(r)=r^2-2M\,r+a^2,
	\end{align*}
	where $M$ and $aM$ denote the mass and the angular momentum of the black hole, respectively.
	As in \cite{FF03} we restrict attention to the non-extreme case $M^2\geq a^2$ and to the region $r>r_1= M +\sqrt{M^2-a^2}$ outside the event horizon, implying that~$\Delta>0$. This metric is of the form~\eqref{eq:metric1b} with
	\[ N = \frac{\Delta\, U}{\sigma^2} \qquad \text{and} \qquad N_i = \frac{2 M r\, a \sin^2 \vartheta}{U} \;(0,0,  1) \]
	and
	\[ \sigma^2 := (r^2+a^2)\, U + 2 M r \,a^2 \sin^2 \vartheta \:. \]
	A short computation shows that Assumption~\ref{assump:met} does not hold in the region \begin{align*}
	r^2-2Mr+a^2\cos^2\vartheta<0 ,
	\end{align*} 
	referred to as the {\em{ergosphere}}. In particular, on the outer ergosurface we have $-N^{2}+N_iN^i=0.$  
	Hence,  we are not able to introduce a weighting factor as simple as before.  Nevertheless, a simple rewriting of the operator $w^2$ supplies the desired result.   First, let us rewrite the operator $w^2$ as follows,
	\begin{align*}
	w^2&=  -N^2(\sqrt{|g|})^{-1}\partial_i\left(   \sqrt{|g|}    h^{ij}     \partial_j \right)+V\\&= 
	-N^2(\sqrt{|g|})^{-1}\partial_i\left(   \sqrt{|g|}    g^{ij}     \partial_j \right)
	+  N^2(\sqrt{|g|})^{-1}\partial_i\left(   \sqrt{|g|}    N^{-2} N^iN^j   \partial_j \right)+V
	\\&=  
	-N^2(\rho\sqrt{|\textbf{g}|})^{-1}
	\partial_i\left(   \rho \,\sqrt{|\textbf{g}|}\,  g^{ij}     \partial_j \right)
	+  N^2(\sqrt{|g|})^{-1}\partial_i\left(   \sqrt{|g|}    N^{-2} N^iN^j  \partial_j \right)+V
	\\&=  
	-N^2\Delta_{\mu,\textbf{g}}
	+  N^2(\sqrt{|g|})^{-1}\partial_i\left(   \sqrt{|g|}    N^{-2} N^iN^j     \partial_j \right)+V,
	\end{align*}
	where we defined the variable $\rho:=(\sqrt{|g|})
	(\sqrt{|\textbf{g}|})^{-1}$ and $\Delta_{\mu,\textbf{g}}$ is a weighted Laplace operator.
	Since the vector $N^i$ only has a $\varphi$-component and the
	components of the Kerr metric do not depend explicitly on $\varphi$, we obtain, 
	\begin{align}\label{eq:w2kerr}
	w^2&=  
	-N^2\Delta_{\mu,\textbf{g}}
	+     N^i N^j     \partial_i\partial_j  +V. 
	\end{align} 
	This operator is symmetric on the Hilbert space~$L^2(\Sigma, d\tilde{\mu})$ with domain~$C^\infty_0(\Sigma)$,
	where,
	\[ d\tilde{\mu} = N^{-2}\, d\mu = N^{-2}\, dr \,d\cos \vartheta \,d\varphi \:. \]
	In order to prove essential self-adjointness, it suffices to consider each $\varphi$-mode
	separately. Indeed, the Hilbert space~$L^2(\Sigma, d\tilde{\mu})$ is the orthogonal direct sum of the
	azimuthal modes, i.e.,
	\begin{equation} L^2(\Sigma, d\tilde{\mu}) = \bigoplus_{k \in \Z} {\mathcal{H}}_k \qquad \text{with} \qquad
	{\mathcal{H}}_k := \big\{ \Phi \in L^2(\Sigma, d\tilde{\mu}) \:\big|\: \Phi(r,\vartheta, \varphi) = e^{-i k \varphi}\: \Xi(r, \vartheta) \big\} \:. \label{eq:secdec} \end{equation}
	Moreover, the operator~$w^2$ is invariant on each of the direct summands. Therefore, the self-adjoint extension of $w^2$ is simply the direct sum of the self-adjoint extensions of~$w^2$ on each~${\mathcal{H}}_k$.
	With this in mind, we may restrict attention to a single azimuthal mode~$k \in \Z$.
	Then for any~$\Phi \in {\mathcal{H}}_k$ as above, we have
	\begin{equation}
w^2\, \Phi =  e^{-i k \varphi}\,{w}^2_k\, e^{i k \varphi}\,  \Phi \:,
\label{eq:relwwk}
\end{equation}
	where~${w}_k^2$ is the operator
	\begin{align}\nonumber
{w}_k^2 &= -N^2\Delta_{\mu,\textbf{g}} -  \frac{1}{4}    \beta^2 +V
	\\&=-\Delta_{\tilde{\mu} ,\tilde{\mathbf{g}}} -\frac{1}{4}   \beta^2   + {V}, \label{eq:wker}
	\end{align}
	with metric  $\tilde{\mathbf{g}}:=N^{-2}{\mathbf{g}}$ and $\beta:=kN^3$ as given in \cite[Equation 2.20]{FF03}. To provide  a proof of essential self-adjointness we first prove that the manifold  $(\Sigma,\tilde{\mathbf{g}})$ is geodesically complete. In order to do so, we supply the following result. 
	\begin{lemma} \label{lemmanewmetric1} 
	The Riemannian manifold $(\Sigma,\hat{\mathbf{g}})$  with  metric
	\begin{equation} \label{newmetric}
	\hat{\mathbf{g}}= \frac{\sigma^2}{\Delta^2} \: dr^2 + \frac{\sigma^2}{\Delta} \big( d\vartheta^2 + \sin^2 \vartheta\: d\varphi^2 \big) \:,
	\end{equation}
	 is geodesically complete. 
	\end{lemma}
	\begin{proof}  
		 Clearly, multiplying the metric by a positive function which is bounded from above and below
		does not change completeness properties. Therefore, instead of Metric~\eqref{newmetric} we can consider the
		warped product metric
		\[ \hat{\mathbf{g}}_N := \frac{r^4}{\Delta^2} \: dr^2 + \frac{r^4}{\Delta} \big( d\vartheta^2 + \sin^2 \vartheta\: d\varphi^2 \big) \:. \]
		The radial part is complete because the function~$r^2/\Delta$ is not integrable both at infinity
		and near the horizon. Applying the result on the completeness of warped product
		metrics~\cite[Lemma~7.2]{bishop} (see also~\citep[page~94]{beem}) completes the proof.
	\end{proof}
 By this lemma we have: 
 	\begin{proposition} \label{lemmanewmetric2} 
 The Riemannian manifold $(\Sigma,\tilde{\mathbf{g}})$ is geodesically complete. 
 \end{proposition}
 \begin{proof} 
 	We   prove that the metric~$\hat{\mathbf{g}}$ 
 is strongly equivalent to  the metric~$\tilde{\mathbf{g}}$ w.r.t.\ the topology induced by the norms. By strong equivalence and Lemma \ref{lemmanewmetric1} geodesic completeness follows.
Strong  equivalence requires proving the following inequalities,
 	\begin{align*}
 	e\,\hat{g}_{ij}u^i\,u^j\leq \tilde{g}_{ij} u^i\,u^j\leq f\,\hat{g}_{ij}u^i\,u^j, 
 	\end{align*}
 	for constants $0<e\leq1$ and $f\geq1$ and tangent vectors $u\in T\Sigma$. Since we have   diagonal metrics with equal  $\tilde{g}_{rr}= \hat{g}_{rr}$,   $\tilde{g}_{\vartheta\vartheta}= \hat{g}_{\vartheta\vartheta}$  components we only need to concentrate on the third component, 
 	\begin{align*}
 	e\, \leq \frac{\sigma^2}{U^2} \leq f.
 	\end{align*} 
 	Since one can rewrite the term ${\sigma^2}/{U^2}$ as 
 	$$1+\frac{a^2\sin^2\vartheta}{U}+\frac{2M\,r\,a^2\sin^2\vartheta}{U^2},$$
 the first inequality with $e=1$ is proven. The second inequality follows from the restriction to the region $r>r_1$ outside the event horizon.  
 \end{proof}

  By using the    Lemma~\ref{lemmanewmetric1} and Proposition~\ref{lemmanewmetric2}) we obtain the following result: 
	\begin{theorem}\label{them:mtad} 
		Let the 	potential $\tilde{V}\in L^2_{\text{loc}}(\Sigma, \tilde{\mu})$ be such that it
		can be written as $\tilde{V} = \tilde{V}_+ + \tilde{V}_-$, where $\tilde{V}_+\in L^2_\text{\text{loc}}(\Sigma , \tilde{\mu})\geq 0$ and $\tilde{V}_-\in L^2_{\text{loc}}(\Sigma, \tilde{\mu})\leq 0$
		point-wise. 	Moreover, let the operator $w^2$ (from Equation~\eqref{eq:w2kerr}) be semi-bounded from below.  Then the operator  $w^2$  is an essentially self-adjoint operator on $C_0^{\infty}(\Sigma )\subset L^{2}(\Sigma,  \,\tilde{\mu})$. 
	\end{theorem}
	\begin{proof} 
	 The proof that ${w}_k^2$  is essentially self-adjoint on $C_0^{\infty}(\Sigma)\subset L^2(\Sigma, d\tilde{\mu})$ follows readily from Proposition \ref{p5}, Theorem \ref{T5}
	and the completeness of the metric $\tilde{\mathbf{g}}$ by Lemma~\ref{lemmanewmetric2}. Since ${w}_k^2$ leaves the sectors $\mathcal{H}_k$ for $k\in\Z$ invariant, recall the decomposition~(\ref{eq:secdec}), it is in particular essentially self-adjoint on each sector. But since the operator $w^2$ is unitary equivalent to ${w}_k^2$ on the sector $\mathcal{H}_k$, recall (\ref{eq:relwwk}), it is essentially self-adjoint there. Combining the sectors we find that ${w}^2$ is essentially self-adjoint on all of $C_0^{\infty}(\Sigma)\subset L^2(\Sigma, d\tilde{\mu})$.
	\end{proof}
	This theorem has the following immediate consequence.
	\begin{corollary}
		Let the potential term $V$ be equal to zero. Then, the spatial derivative  operator  $w^2$
		is an essentially self-adjoint operator on $C_0^{\infty}(\Sigma )\subset L^{2}(\Sigma,  \,\tilde{\mu})$. 
	\end{corollary}
	\begin{proof}
		For the case $V=0$ it follows readily from the bound $\vert\beta\vert\leq c$, that the operator $w^2$ is semi-bounded (since $-\Delta_{\tilde{\mu} ,\tilde{\mathbf{g}}}$ is a positive operator). Hence, Theorem \ref{them:mtad} concludes the proof.
	\end{proof}
 
	\subsection{Non-Globally Hyperbolic Spacetimes}
	
	The advantage of studying the class of global hyperbolic spacetimes $M=\R\times\Sigma$, where $\Sigma$ is a Cauchy surface for the Klein-Gordon equation, is   having  a well-posed  Cauchy problem. In particular, by Leray's Theorem~\cite{LJ}, the initial data  on the Cauchy surface $\Sigma$ define a unique solution of the Klein-Gordon equation.  However, there are interesting \textit{non-globally hyperbolic} spacetimes that arise as solutions in general relativity, e.g., the G\"odel-Universe or the anti-de Sitter space, see \cite{HE73}. The unaesthetic feature of such spacetimes is that by definition, there is no initial data surface that is a Cauchy surface. Thus, in such spacetimes the equations of motion cannot predict from initial conditions the outcome of the dynamics in certain regions of this spacetime.
	
	An important step towards addressing the Cauchy problem in certain classes of static non-globally hyperbolic  spacetimes was proposed by \cite{NGH2}, see also \cite{NGH1} in this context. Principally,  in \cite{NGH2} a (unique, see \cite{NGH3}) prescription  of defining the dynamics was given. The main ingredient is finding self-adjoint extensions of the spatial part of the Klein-Gordon equation (i.e.\ $w^2$). Self-adjointness of the operator $w^2$ is proven by using its symmetry and positivity and hence by being able to apply the Friedrichs extension that guarantees the existence of at least one self-adjoint extension. 
	Moreover, in \cite{NGH3} the authors were able to prove the uniqueness of this method and in   \cite{NGH4} the construction was extended to stationary non-globally hyperbolic spacetimes.
	
	In this section we consider non-globally hyperbolic static and stationary spacetimes and study the essential self-adjointness of the operator $w^2$.   
	Let us consider the non-globally hyperbolic spacetime manifold $(M,g)$, where $M=\R\times\Sigma$, with a metric  $g$ given by
	\begin{align*}
	g=-N^2(\vec{x}) dt^2+g_{ij}(\vec{x})\,(dx^i+N^i(\vec{x})\,dt)\,(dx^j+N^j(\vec{x})\,dt),
	\end{align*}
	where $(\Sigma, \textbf{g})$ is a Riemannian manifold and by definition not a Cauchy-surface. Moreover, let 
	the shift function be bounded from below and above (by constants $\alpha_{B}, \alpha_{C}\in\mathbb{R}$). In addition 
	Assumption~\ref{assump:met} holds, which implies $N_iN^i<\alpha_{C}$. Then, by Theorem \ref{thm:ghst} it is clear that for such a non-globally hyperbolic spacetime $M$ the Riemannian manifold $(\Sigma,\mathbf{g})$ is not complete.  Since conformal transformations preserve (non-)global hyperbolicity it is clear that the Riemannian manifold $(\Sigma,\tilde{g};=N^{-2}\mathbf{g})$ is not complete either. However, the following result holds.
	\begin{theorem} \label{thm54}
		Let $(M,g)$ be a non-globally hyperbolic spacetime, where $M=\R\times\Sigma$, with   metric   
		\begin{align*}
		g=-N^2(\vec{x}) dt^2+g_{ij}(\vec{x})\,(dx^i+N^i(\vec{x})\,dt)\,(dx^j+N^j(\vec{x})\,dt).
		\end{align*}
		Moreover,  let $(\Sigma, \textbf{g})$ be a smooth Riemannian manifold,   
		the shift function be bounded from below and above (by constants $\alpha_{B}, \alpha_{C}\in\mathbb{R}$)
		and let the shift vector $\vec{N}$ be the derivative of a   smooth proper function on $\Sigma$, with   bound $\vec{N}^2<N^2$.  Then, the Riemannian manifold $(\Sigma,\mathbf{\tilde{h}})$, where $\tilde{h}_{ij}= N^{-2}{g}_{ij}+ (1-  N^{-2}\vec{N}^2)^{-1}N^{-4}N_iN_j$,  is   geodesically complete. 
	\end{theorem} \begin{proof}  
		From the non-global hyperbolicity of the manifold $(M,g)$ it follows with Theorem \ref{thm:ghst} that the Riemannian manifold $(\Sigma,\textbf{g})$ is not complete.  However, since the shift vector is a derivative of a smooth proper function (on $\Sigma$), it follows by \cite[Theorem 1]{WG} that the Riemannian manifold $(\Sigma,\textbf{k})$ with metric $\textbf{k}$  given by  
		\begin{align*}
		k_{ij}:= N^2\,g_{ij}+ N_i N_j,
		\end{align*}
		is complete. Next, let us define the Riemannian metric $\tilde{\mathbf{k}}$ as
		\begin{align*}
		\tilde{k}_{ij}:= g_{ij}+ N^{-2} N_iN_j.
		\end{align*}
		From the boundedness of $N^2$ it follows that    the norms of $\tilde{\mathbf{k}}$ and ${\mathbf{k}}$ are strongly equivalent, i.e.,
		\begin{align*}
		\alpha_{B} \tilde{k}_{ij}u^iu^j\leq {k}_{ij}u^iu^j\leq \alpha_{C} \tilde{k}_{ij}.
		\end{align*}
		Hence, from the completeness of the Riemannian manifold $(\Sigma,\textbf{k})$ it follows that $(\Sigma,\tilde{\mathbf{k}})$ is geodesically complete as well.  Next,  we turn to  the Riemannian metric $\textbf{h}$ given as
		\begin{align*}
		h_{ij}= g_{ij}+N^{-2}(1-\Vert\vec{N}\Vert^2_{\tilde{\mathbf{g}}})^{-1}N_iN_j,
		\end{align*}
		with $\tilde{\mathbf{g}}:= N^{-2}\mathbf{g}$. Since $\tilde{\mathbf{k}}$ is complete, $\textbf{h}>0$ and $\textbf{h}-\tilde{\mathbf{k}}>0$  for any tangent vectors $u\in T\Sigma$,
		\begin{align*}
		(h_{ij}-k_{ij})u^iu^j=( (1-  \Vert\vec{N}\Vert^2_{\tilde{\mathbf{g}}})^{-1}N^{-2}N_iN_j-N^{-2}N_iN_j )u^iu^j
		>0.
		\end{align*}
		This is satisfied by multiplying with the term $(1-  \Vert\vec{N}\Vert^2_{\tilde{\mathbf{g}}})>0$, 
		\begin{equation*}
		(  1-(1-  \Vert\vec{N}\Vert^2_{\tilde{\mathbf{g}}}  )N^{-2} N_iN_ju^iu^j =
		(  \Vert\vec{N}\Vert^2_{\tilde{\mathbf{g}}}  )N^{-2}N_iN_j\,u^iu^j
		>0.
		\end{equation*} 
		It follows  that  the Riemannian manifold $(\Sigma,\textbf{h})$ is complete, (see \cite[Remark 2]{WG}).  
		The focus of our interest is   the conformally transformed Riemannian manifold,  
		$(\Sigma,\tilde{\textbf{h}})$, with the metric  $\tilde{\textbf{h}}=N^{-2}\textbf{h}$. Since $N^{2}$ is bounded from below and above it follows by the strong equivalence of the norms of the metrics $\tilde{\textbf{h}}$ and ${\textbf{h}}$, i.e.,
		\begin{align*}
		\alpha_{B} \tilde{h}_{ij}u^iu^j\leq {h}_{ij}u^iu^j\leq \alpha_{C} \tilde{h}_{ij},
		\end{align*}
		for any tangent vectors $u\in T\Sigma$ that $(\Sigma,\tilde{\textbf{h}})$ is geodesically complete.	
	\end{proof}
	Therefore, the operator  $w^2$ (from Equation~\ref{op}) is an \textbf{essentially} self-adjoint operator on $C_0^{\infty}(\Sigma )\subset L^{2}(\Sigma,\tilde{\textbf{h}},  {\mu})$. Next, let us consider a static spacetime   $(M,g)$, where $M=\R\times\Sigma$, with a metric  $g$ given by
	\begin{align*}
	g=-N^2(\vec{x})\,dt^2+g_{ij}(\vec{x})\,dx^i\,dx^j,
	\end{align*}
	where $\mathbf{g}$ is the induced Riemannian metric on the smooth manifold $\Sigma$.
	For a non-globally hyperbolic spacetime $M$ the Riemannian manifold   $(\Sigma,N^{-2}\mathbf{g})$ is not complete. Since if it were complete the spacetime $(M,\tilde{g})$, where $\tilde{g}=N^{-2}g$, would be globally hyperbolic and so would the spacetime $(M, {g})$.\footnote{ 
		Moreover if $N^2(\vec{x})$ is bounded  it follows from Theorem \ref{thm:ghst}  (or \cite[Proposition 5.2]{K78})   that the Riemannian manifold   $(\Sigma, \mathbf{g})$ is not complete.} 
	Therefore,   \cite[Theorem 4.1]{MO} regarding essential self-adjointness of the operator $w^2$ does not apply.  However, one can complete the Riemannian manifold $(\Sigma,N^{-2}\mathbf{g})$ by using the work of \cite{WG} in a natural way.  
	\begin{theorem} \label{thm55}
		Let  $(M,g)$ be a non-globally hyperbolic static spacetime, where $M=\R\times\Sigma$, with   metric  $g$ given by
		\begin{align*}
		g=-N^2(\vec{x})\,dt^2+g_{ij}(\vec{x})\,dx^i\,dx^j,
		\end{align*}
	and let $(\Sigma, \textbf{g})$ be a smooth Riemannian manifold and $-N^2+N_iN^i<0$. Then, there exists a  vector $\vec{\gamma}$   such that  $\vec{\nabla}{\gamma}=\vec{\gamma}$, where $\gamma$ is a smooth  proper function on $\Sigma$ and  it follows that  
		the Riemannian manifold $(\Sigma,N^{-2}e^{\Vert\vec{\gamma}\Vert^2_{\textbf{g}}}\mathbf{g})$ is complete. Hence, the  warped product manifold $(M_{f},g_f)$, where $M_{f}=\Sigma\times_f\R$ with $f=e^{-\Vert\vec{\gamma}\Vert^2_{\textbf{g}}}, \,\,f :\Sigma\rightarrow(0,\infty)$ and metric
		\begin{align*}
		g_f=-e^{-\Vert\vec{\gamma}\Vert^2_{\textbf{g}}}N^2(\vec{x}) dt^2+g_{ij}(\vec{x})\, dx^i \, dx^j,
		\end{align*} 
		is  globally hyperbolic. 
	\end{theorem}
	\begin{proof}
		The existence of   a smooth proper function $\gamma$ on  a smooth manifold follows from \cite[Proposition 1.3.5]{Pet}. From \cite[Theorem 1]{NOZ}, \cite[Corollary]{WG} it follows that $(\Sigma,N^{-2} e^{\Vert\vec{\gamma}\Vert^2_{\textbf{g}}}\mathbf{g})$ is   complete and hence   \cite[Proposition 5.2]{K78}  insures the global hyperbolicity of 
		$(M_{f},N^{-2} e^{\Vert\vec{\gamma}\Vert^2_{\textbf{g}}}g_f)$. Since conformal transformations preserve the causal structure it follows that  	$(M_{f}, g_f)$ is globally hyperbolic as well. 
	\end{proof}  
	There is a straightforward generalization of this result to the case of stationary spacetimes. 
	\begin{theorem} Let  $(M,g)$ be a non-globally hyperbolic stationary spacetime, where $M=\R\times\Sigma$, with   metric  $g$ given by
		\begin{align*}
		g=-N^2(\vec{x})\,dt^2+g_{ij}(\vec{x})\,(dx^i+N^i(\vec{x})\,dt)\,(dx^j+N^j(\vec{x})\,dt),
		\end{align*}
and let $(\Sigma, \textbf{g})$ be a smooth Riemannian manifold. Then, there exists a  vector $\vec{\gamma}$   such that  $\vec{\nabla}{\gamma}=\vec{\gamma}$, where $\gamma$ is a smooth  proper function on $\Sigma$ and  it follows that  
		the Riemannian manifold $(\Sigma,N^{-2}e^{\Vert\vec{\gamma}\Vert^2_{\textbf{g}}}\mathbf{g})$ is complete. Hence, the  warped product manifold $(M_{f},g_f)$, where $M_{f}=\Sigma\times_f\R$ with $f=e^{-\Vert\vec{\gamma}\Vert^2_{\textbf{g}}}, \,\,f :\Sigma\rightarrow(0,\infty)$ and metric
		\begin{align*}
		g_f=-e^{-\Vert\vec{\gamma}\Vert^2_{\textbf{g}}}N^2(\vec{x}) dt^2+g_{ij}(\vec{x})\,(dx^i+N^i(\vec{x})\,dt)\,(dx^j+N^j(\vec{x})\,dt),
		\end{align*} 
		is  globally hyperbolic. 
	\end{theorem}
	\begin{proof}
		The existence of   a smooth proper function $\gamma$ on  a smooth manifold follows from \cite[Proposition 1.3.5]{Pet}. From \cite[Theorem 1]{NOZ}, \cite[Corollary]{WG} it follows that $(\Sigma, {\tilde{\mathbf{g}}})$, where $\tilde{\mathbf{g}}:=
		N^{-2} e^{\Vert\vec{\gamma}\Vert^2_{\textbf{g}}}\mathbf{g})$, is   complete. Moreover, since the scalar product of the shift vectors $\beta^i(\vec{x},t)$ for the spatial metric $N^{-2} e^{\Vert\vec{\gamma}\Vert^2_{\textbf{g}}}\mathbf{g}$ is uniformly bounded from above, it follows from Theorem \ref{thm:ghst} that 	$(M_{f},N^{-2} e^{\Vert\vec{\gamma}\Vert^2_{\textbf{g}}}g_f)$
		is globally hyperbolic. As in the previous proof, we use the fact that conformal transformations preserve the causal structure and hence it follows that  	$(M_{f}, g_f)$ is globally hyperbolic as well. 
	\end{proof}  
	From this Theorem we obtain the following result.
	\begin{corollary} Let  the
		metric $\tilde{\mathbf{h}}$ be given by  
		$$\tilde{h}_{ij}=\tilde{g}_{ij}+ (1-  N^{-2}e^{\Vert\vec{\gamma}\Vert^2_{\textbf{g}}}\vec{N}^2)^{-1}N^{-4}e^{\Vert\vec{\gamma}\Vert^4_{\textbf{g}}}N_iN_j,$$
		where  $\tilde{\mathbf{g}} =
		N^{-2} e^{\Vert\vec{\gamma}\Vert^2_{\textbf{g}}}\mathbf{g}$ is the Riemannian metric from the former theorem and let $N^{-2}e^{\Vert\vec{\gamma}\Vert^2_{\textbf{g}}}\vec{N}^2<1$. Then, the Riemannian manifold $(\Sigma, \tilde{\mathbf{h}})$ is  complete. 
	\end{corollary}
	\begin{proof}We consider the case $\vec{N}\neq0$, due to the fact that   the case  $\vec{N}=0$ follows   from the previous theorem.  Since $\tilde{\mathbf{g}}$ is complete, $\tilde{\mathbf{h}}>0$ and     $\tilde{\mathbf{h}}-\tilde{\mathbf{g}}>0$  for any tangent vectors $u\in T\Sigma$,
		it follows  that  the Riemannian manifold $(\Sigma,\textbf{h})$ is complete, (see \cite[Remark 2]{WG}).
	\end{proof}
	
	\section*{Acknowledgments}
	A.~Much and R.~Oeckl acknowledge partial support from CONACYT project 259258. A.~Much acknowledges  support from Fordecyt project 265667.
	
	\appendix
	\section{Proofs}
	\subsection{Proof of Lemma \ref{p1}}\label{pp1}
	
	\begin{proof}	
		In order to see the explicit form of the Klein Gordon Equation we first use the splitting of a spacetime provided by  the metric and the inverse metric  and hence  a straight forward calculation yields, 
		\begin{align*} &
		\left(({\sqrt{|g|}})^{-1}\partial_{\mu}(\sqrt{|g|}g^{{\mu}{\nu}}\partial_{\nu})\right)\phi \\=&  
		\left(	g^{00} \partial_0^2 +(\sqrt{|g|})^{-1}\partial_i(   \sqrt{|g|}   g^{i0}    \partial_0 )
		+(\sqrt{|g|})^{-1}\partial_0(   \sqrt{|g|}   g^{0i}    \partial_i )
		+ (\sqrt{|g|})^{-1}\partial_i(   \sqrt{|g|}    h^{ij}     \partial_j )\right)\phi\\=&  
		\left(		g^{00} \partial_0^2 +(\sqrt{|g|})^{-1}\partial_i(   \sqrt{|g|}   N^{-2}N^i    \partial_0 ) 
		+(\sqrt{|g|})^{-1}\partial_0(   \sqrt{|g|}   N^{-2}N^i    \partial_i ) 
		+ (\sqrt{|g|})^{-1}\partial_i(   \sqrt{|g|}    h^{ij}     \partial_j )\right)\phi\\= &  
		\left(	g^{00} \partial_0^2 +(\sqrt{|g|})^{-1}\partial_i(   \sqrt{|g|}   N^{-2}N^i    \partial_0 ) 
		+ N^{-2}N^i    \partial_i \partial_0
		+ (\sqrt{|g|})^{-1}\partial_i(   \sqrt{|g|}    h^{ij}     \partial_j )\right)\phi,
		\end{align*}
		which in turn  leads to  
		\begin{align*} 
		&   
		\left(	 \partial_0^2  + f    \partial_{0}  +w^2\right)\phi=
		0 .
		\end{align*} 
		Here $f$ is defined as  
		\begin{align*}
		f&:= -  (g^{00}\,\sqrt{|g|})^{-1} \partial_i(    \sqrt{|g|}\,  g^{00} \, N^{i} ) -  2 N^i\partial_{i},
		\end{align*}   
		where the operator $w$ is given by 
		\begin{align*}
		w^2&=  (g^{00})^{-1}(\sqrt{|g|})^{-1}\partial_i\left(   \sqrt{|g|}    h^{ij}     \partial_j \right)-(g^{00})^{-1}m^2(x)\\&=
		+ (g^{00})^{-1}(\sqrt{|g|})^{-1}(\sqrt{ |\textbf{h}|})(\sqrt{ |\textbf{h}|})^{-1} \partial_i\left(   \sqrt{|g|}  (\sqrt{ |\textbf{h}|})^{-1}(\sqrt{ |\textbf{h}|})   h^{ij}     \partial_j \right)-(g^{00})^{-1}m^2(x).
		\end{align*}

	\end{proof}

	\subsection{Proof of Lemma \ref{lem:rho}}\label{plem:rho}
	\begin{proof}  
		First, we calculate the explicit form of the function $\rho$ by calculating the ration of  the determinants $|g|$ and $|\textbf{h}|$. We begin with $|g|$,
		\begin{align*}
		|g|&= |
		N_3^2 g_{12} g_{21} - N_2 N_3 g_{13} g_{21} - N_3^2 g_{11} g_{22} + N_1 N_3 g_{13} g_{22} + N_2 N_3 g_{11} g_{23} - N_1 N_3 g_{12} g_{23} \\&- N_2 N_3 g_{12} g_{31} + N_2^2 g_{13} g_{31} + N_1 N_3 g_{22} g_{31} - N_1N_2 g_{23} g_{31} + N_2 N_3 g_{11} g_{32} - N_1N_2 g_{13} g_{32}  \\&- N_1 N_3 g_{21} g_{32}+ N_1^2 g_{23} g_{32} - N_2^2 g_{11} g_{33} + N_1N_2 g_{12} g_{33} + N_1N_2 g_{21} g_{33} - N_1^2 g_{22} g_{33}  +  g_{00} \det(\textbf{g})| .
		\end{align*}
		Here, the determinant of the spatial part of the metric $g$ is 
		\begin{align*}
		\det(\textbf{g}) &= -g_{13} g_{22} g_{31} + g_{12} g_{23} g_{31} + g_{13} g_{21} g_{32} - g_{11} g_{23} g_{32} - g_{12} g_{21} g_{33} + g_{11} g_{22} g_{33} ,
		\end{align*}
		while the determinant of the metric $h_{ij}$ is given as
		\begin{align*}
		|\textbf{h}|&= |g_{00}^{-1}| | N_1^2 g_{23} g_{32} - N_1^2 g_{22} g_{33} + N_3 N_1 g_{13} g_{22} - N_3 N_1 g_{12} g_{23} + N_3 N_1 g_{22} g_{31} - N_2 N_1 g_{23} g_{31}  \\&-N_2 N_1 g_{13} g_{32}- N_3 N_1 g_{21} g_{32}  + N_2 N_1 g_{12} g_{33} + N_2 N_1 g_{21} g_{33} + N_3^2 g_{12} g_{21} - N_2 N_3 g_{13} g_{21}\\& - N_3^2 g_{11} g_{22} + N_2 N_3 g_{11} g_{23} - N_2 N_3 g_{12} g_{31} + N_2^2 g_{13} g_{31} + N_2 N_3 g_{11} g_{32} - N_2^2 g_{11} g_{33} +  g_{00} \det(\textbf{g})|.
		\end{align*}
		Hence, it follows readily that, 
		\begin{align*}
		|\textbf{h}|=|g_{00}^{-1}||{g}| .
		\end{align*} While positivity is obvious, smoothness of the function $\rho$ follows from Assumption \ref{assump:met}.
	\end{proof}
	
	\section{Supplementary Material}
	
	\subsection{Shubin's Theorem}\label{thm:sch}
The following result \cite[Theorem 1.1]{S01} is fundamental for many proofs in this work. 
	\begin{theorem}\label{T5}
		Let the Riemannian manifold $(\Sigma,\mathbf{h})$ be complete and let the potential  $V\in L^2_{loc}(\Sigma, {\nu})$ be such that we can write $V = V_+ + V_-$, where $V_+\in L^2_{loc}(\Sigma, {\nu})\geq 0$ and $V_-\in L^2_{loc}(\Sigma, {\nu})\leq 0$
		point-wise and the corresponding operator  
		$-{\Delta}_{{\nu}}+V$
		be 	semi-bounded from below. Then, the operator    $-{\Delta}_{{\nu}}+V$ is an essentially self-adjoint operator on $C_0^{\infty}(\Sigma,{\nu})$.
	\end{theorem}
	We write $f\in L^2_{loc}(\Sigma, {\nu})$  for 
	a local $L^2(\Sigma, {\nu})$ function $f$  that is an element of the weighted Hilbert space $L^2(\Sigma, {\nu})$ on every compact subset of the manifold $\Sigma$. 
	

\end{document}